\documentclass[12pt]{amsart}
\usepackage[all]{xy}
\usepackage{amsfonts,amsmath,oldgerm,amssymb,amscd}
\usepackage{enumerate}
\usepackage{makecell}
\UseComputerModernTips
\numberwithin{equation}{section}
\usepackage[breaklinks]{hyperref}

\usepackage{float}
\usepackage{xcolor}
\makeatletter
\@namedef{subjclassname@2010}{%
\textup{2010} Mathematics Subject Classification}
\makeatother
\usepackage{bm}
\usepackage{caption}
\usepackage{bm}

\newtheorem{theorem}{Theorem}[section]
\newtheorem{lemma}[theorem]{Lemma}
\newtheorem{proposition}[theorem]{Proposition}
\newtheorem{corollary}[theorem]{Corollary}
\newtheorem{definition}[theorem]{Definition}
\newtheorem{remark}[theorem]{Remark}
\newtheorem{example}[theorem]{Example}

\title{ Code Size Constraints in $b$-symbol Read Channels: A Bound Analysis}
\author{ Gyanendra K. Verma}
\address{Department of Mathematical Sciences, United Arab Emirates University, Al Ain, 15551, UAE.}
\email{gkvermaiitdmaths@gmail.com}
\author{Nupur Patanker}
\address{Indian Institute of Science Bangalore, India}
\email{ nupurp@iisc.ac.in}
\author{Abhay Kumar Singh$^*$}
\address{Department of Mathematics and Computing, Indian Institute of
Technology (ISM), Dhanbad, India.}
\email{abhay@iitism.ac.in;}
\thanks{
*Corresponding author. Partial work was done while G. K. Verma was affiliated with the Department of Mathematics, Indian Institute of Technology Delhi. G. K. Verma is financially supported by UAEU grant G00004614. N. Patanker was supported in part by an IoE-IISc Postdoctoral Fellowship awarded by the Indian Institute of Science Bangalore, and in part by the grant DST/INT/RUS/RSF/P-41/2021 from the Department of Science \& Technology, Govt. of India.
}

\usepackage{fancyhdr}
\begin{document}
\date{}
\maketitle

\begin{abstract}
In classical coding theory, error-correcting codes are designed to protect against errors occurring at individual symbol positions in a codeword. However, in practical storage and communication systems, errors often affect multiple adjacent symbols rather than single symbols independently. To address this, symbol-pair read channels were introduced \cite{Yuval2011}, and later generalized to $b$-symbol read channels \cite{yaakobi2016} to better model such error patterns. $b$-Symbol read channels generalize symbol-pair read channels to account for clustered errors in modern storage and communication systems. By developing bounds and efficient codes, researchers improve data reliability in applications such as storage devices, wireless networks, and DNA-based storage. Given integers $q$, $n$, $d$, and $b \geq 2$, let $A_b(n,d,q)$ denote the largest possible code size for which there exists a
$q$-ary code of length $n$ with minimum $b$-symbol distance at least $d$. In \cite{chen2022}, various upper and lower bounds on $A_b(n,d,q)$ are given for $b=2$. In this paper, we generalize some of these bounds to the $b$-symbol read channels for $b>2$ and present several new bounds on $A_b(n,d,q)$. In particular, we establish the linear programming bound, a recurrence relation on $A_b(n,d,q)$, the Johnson bound (even), the restricted Johnson bound, the Gilbert-Varshamov-type bound, and the Elias bound for the metric of symbols $b$, $b\geq 2$. Furthermore, we provide examples demonstrating that the Gilbert-Varshamov bound we establish offers a stronger lower bound than the one presented in \cite{Song2018}. Additionally, we introduce an alternative approach to deriving the Sphere-packing and Plotkin bounds.

\end{abstract}

\section{Introduction}
In coding theory, the function 
$A_H(n,d,q)$  represents the maximum number of codewords in a code of length $n$ over an alphabet of size 
$q$ with a minimum Hamming distance of at least 
$d$. Determining its exact value is a challenging problem, but various upper and lower bounds provide useful estimates. The Singleton Bound establishes a simple upper limit, while the Sphere Packing Bound and Plotkin Bound offer more refined constraints based on geometric considerations. The Linear Programming Bound provides even tighter estimates using optimization techniques. On the lower end, the Gilbert-Varshamov Bound guarantees the existence of codes with a certain minimum size, and the Griesmer Bound helps in understanding the length of linear codes. These bounds play a crucial role in designing efficient error-correcting codes, balancing redundancy and error detection capability to optimize communication systems.\par 
The $b$-symbol distance is a generalization of the Hamming distance that provides a more refined measure of differences between codewords, particularly in scenarios where errors tend to occur in bursts rather than independently at individual positions. In this metric, rather than counting all differing symbols between two codewords, the distance measure considers contiguous blocks of up to 
$b$-symbol as single units.
 Cassuto and Blaum \cite{Yuval2011} introduced symbol-pair distance (i.e. $2$-symbol distance) which is further generalized to $b$-symbol distance for $b>2$ by Yaakobi et al. \cite{yaakobi2016}. 
These generalizations are motivated by the limitations of traditional high-density data storage systems, which do not permit the retrieval of individual symbols from the channel. The $b$-symbol coding framework enables the simultaneous reading of overlapping $b$-symbols instead of processing them separately. This innovative approach has garnered significant attention, leading to substantial progress in the study of $b$-symbol codes (see \cite{chee2013,ding2018,Dinh2018,Elishco2020,Luo2023}). Similar to Hamming distance, for the $b$-symbol distance, we have a quantity $A_b(n,d,q)$ that for given integers $n$ and $d$ with $1 \leq d \leq n$, denotes the maximum number of codewords in a code over an alphabet of size $q$ of length $n$ and  $b$-symbol distance at least $d$. Similar to its counterparts in the Hamming metric, the quantity $A_b(n,d,q)$ plays an important role, and various bounds on its value have been determined. The $b$-symbol Sphere Packing Bound and Gilbert-Varshamov bound has been determined in \cite{Yuval2011} for $b=2$ and in \cite{Song2018}, the authors established Sphere Packing Bound and Gilbert–Varshamov Bound for $b>2$. The $b$-symbol Singleton bound has been studied in \cite{chee2013} for $b=2$ and in \cite{ding2018} for $b> 2$. The Johnson Bound has been determined in \cite{Elishco2020} for even minimum pair-distance (binary alphabet) and in \cite{chen2022} for arbitrary minimum pair-distance. Linear Programming bound for $b=2$ has been derived in \cite{Elishco2020}, providing an optimization-based upper bound on code size. The Griesmer Bound, established in \cite{Luo2024}, offers a lower bound on the length of linear codes given a $b$-symbol distance. The Sphere-Covering Bound, determined in \cite{chen2022}, evaluates how efficiently codewords can cover the space while maintaining the required $b$-symbol distance, while the Plotkin Bound, also derived in \cite{chen2022}, provides an upper limit on code size when the $b$-symbol distance is relatively large.
Asymptotic Bounds on code size have been studied in \cite{Yuval2011}, \cite{Yuval2}, \cite{Liu2019}. Chen \cite{Chen2025} studied covering codes and covering radii over $b$-symbol read channels. A code is considered optimal if it meets a given bound. In particular, when the minimum $b$-symbol distance of a code achieves the Singleton bound, it is referred to as a $b$-symbol maximum distance separable (MDS) code. Various constructions of $b$-symbol codes attaining specific bounds have been explored in \cite{chee2013,Chen2025,ding2018,Luo2023,Luo2024,yaakobi2016,Zhao2025}. 

In this paper, we extend the study of $b$-symbol codes by establishing several new upper and lower bounds on $A_b(n,d,q)$. While some of these bounds are entirely new, others serve as generalizations of existing results that were previously known only for the special case of 
$b=2$. By broadening these bounds to arbitrary values of 
$b$, we provide a deeper understanding of the structure and limitations of 
$b$-symbol codes, which are particularly relevant for coding scenarios where errors tend to occur in bursts rather than independently at individual positions. A brief list of known bounds and the contribution of this paper is listed in Table \ref{tab1}. The remainder of the paper is structured as follows: Section \ref{pre} provides fundamental definitions and a review of known results. The subsequent Section \ref{boundoncode} establishes a key recurrence relation on $A_b(n,d,q)$, providing a recursive framework. In Section \ref{additivebound}, we derive several bounds on $A_b(n,d,q)$ through an algebraic approach. In Section \ref{slpbound}, we present the linear programming bound for $b$-symbol codes.
\begin{table}[]
    \centering
    \begin{tabular}{|c|c|c|}
    \hline
       \textbf{Bound}  & $b=2$& $b\geq 2$ \\
         \hline
       \text{Sphere-packing bound}  &  \cite{Yuval2011}  & \cite{Song2018}, \textbf{New$^{\#}$}, Corollary \ref{spb}   \\
       \hline
       \text{Gilbert–Varshamov bound}& \cite{Yuval2011} &\cite{Song2018}, \textbf{New$^{*}$}, Corollary \ref{Gilbert}  \\
       \hline
       \text{Singleton bound}  & \cite{chee2013}  & \cite{ding2018} \\
       \hline
       \text{Johnson bound }  & \cite{Elishco2020,chen2022} & \textbf{ New}, Theorem \ref{jboundeven} \\
       \hline
       \text{Linear Programming bound}  & \cite{Elishco2020} & \textbf{New}, Theorem \ref{lpbound} \\
       \hline 
       \text{Recurrence relation}  & \cite{chen2022} & \textbf{New}, Theorem \ref{abnqd}\\
       \hline
       \text{Plotkin bound}  &  \cite{Wang2020} & \cite{yang2016}, \textbf{New$^{\#}$}, Corollary \ref{plotb} \\
       \hline
        \text{Restricted Johnson bound}  & \cite{Wang2020} & \textbf{New}, Corollary \ref{rjbound} \\
        \hline
     \text{Griesmer bound }  & \cite{Luo2024} &  \cite{Luo2024}\\
        \hline
        
        \text{Elias bound }  & \cite{Yan2020} & \textbf{New}, Corollary \ref{elias} \\
        \hline
    \end{tabular}
    \caption{Bounds on $A_b(n,d,q)$ for $b$-symbol distance. 
    $\#$ : we provide an alternate approach to derive the bounds, $*$ : we give an improved version of the bound.}
    \label{tab1}
\end{table}



\section{Preliminaries}\label{pre}

Throughout this paper by $\mathbb{F}_{q}$ we denote a finite field with $q$ elements where $q$ is a power of prime $p$. A code $C$ over $\mathbb{F}_{q}$ of length $n$ is a non-empty subset $C$ of $\mathbb{F}_q^n$ and its elements are known as codewords. The Hamming weight of an element $x$  in $\mathbb{F}_q^n$ is the number of nonzero coordinates in $x$, denoted as $w_H(x)$. The Hamming distance between two vectors $x$ and  $y$  in $\mathbb{F}_q^n$ is given by the Hamming weight of their difference, $w_H(x-y)$. The minimum distance $d_H$  of $C$ is defined as the minimum Hamming distance between any two distinct codewords in $C$. \par
For an integer $1\leq b\leq n$, define a map $\pi: \mathbb{F}_q^n \to (\mathbb{F}_q^b)^n$ given by
$$\pi(y_1,y_2,\dots,y_n)=((y_1,y_2,\dots,y_b),(y_2,y_3,\dots,y_{b+1}),\dots,(y_n,y_1,\dots,y_{b-1})).$$
\begin{definition}[$b$-symbol weight]
For $y\in \mathbb{F}_q^n,$ $b$-symbol weight of $y$ is defined as
\begin{align*}
    w_b(y)=&w_H(\pi(y))=w_H(((y_1,y_2,\dots,y_b),(y_2,y_3,\dots,y_{b+1}),\dots,(y_n,y_1,\dots,y_{b-1})))\\
    =&|\{1\leq i\leq n : (y_i,y_{i+1},\dots,y_{i+b-1})\neq (0,0,\dots,0)\}|,
\end{align*}
 \text{ where indices are taken}$\pmod{n}$.
\end{definition}
The $b$-symbol distance between two vectors $x,y\in \mathbb{F}_q^n$ is defined as $d_b(x,y)=w_H(\pi(x)-\pi(y))$.  The minimum $b$-symbol distance of $C$ is defined as 
$$d_b(C)=min\{d_b(x,y)| \ x\neq y,\  x,y \in C\}.$$
We denote $(n,M,d_b)_q$, the parameters of a code $C$ over $\mathbb{F}_q$ of length $n$ and size $M$ with minimum $b$-symbol distance $d_b.$ It is easy to see  $w_1(x)=w_H(x)$ and $d_1(x,y)=d_H(x,y)$. 
$$wt_H(y) + b -1 \leq  wt_b(y) \leq  b \cdot wt_H(y).$$

$$d_b(C)\leq min\{n-k+b,n\}.$$
We define a $b$-ball of radius $r$ with center $a\in \mathbb{F}_q^n$ as
$$B_b(n,r,a)=\{y\in \mathbb{F}_q^n |\  d_b(y,a)\leq r\}$$ 
and a $b$-sphere of radius $r$ with center $a\in \mathbb{F}_q^n$ as
$$S_b(n,r,a)=\{y\in \mathbb{F}_q^n| \  d_b(y,a)=r\}.$$ In the $b$-symbol distance, the cardinality of a $b$-ball and a $b$-sphere is independent of the choice of the center. We denote $B_b(n,r)=|B_b(n,r,a)|$ and $S_b(n,r)=|S_b(n,r,a)|$ for some $a\in \mathbb{F}_q^n.$  The expressions for $B_b(n,r)$ and $S_b(n,r)$ have been explicitly derived for the cases $b=2$ and $b=3$ in \cite{Elishco2020, Singh2025}. For arbitrary $b$, a general formula was presented in \cite{Song2018}.  
We denote  $$A_b(n,d,q)=max\{M|\  \text{there exists a } q\text{-ary } (n,M) \text{ code } C  \text{ with } d_b(C)\geq d\}$$ and 
\begin{align*}
  A_b(n,d,w,q)=max&\{M| \ \text{there exist a constant weight }  q\text{-ary } (n,M)  \text{ code } C \\
  &\text{ with } d_b(C)\geq d \text{ and } w_b(c)=w\ \forall c\in C\}.  
\end{align*}
For $b=1,$ we denote $A_1(n,d,q)$ as $A_H(n,d,q)$ and $A_1(n,d,w,q)$ as $A_H(n,d,w,q)$.

\section{Recurrence relation on codes size for $\mathbf{b}$-symbol read channels}\label{boundoncode}
In \cite{chen2022}, the authors provided several bounds on $A_2(n,d,q).$ In this section, we generalize their results for $A_b(n,d,q)$, $b \geq 2$. We begin this section by presenting a few key lemmas that will be essential in the proof of our main result.
\begin{lemma}\label{3b-1} 
Let $l,b$ be integers such that $n \geq b\geq 2$ and $n \geq l\geq b$. For $z=(z_1,z_2,\dots,z_n)\in \mathbb{F}_q^n$, let
    \begin{align*}
    B_1=&((z_{l-b+2},\dots,z_{l+1}),\dots,(z_{l},z_{l+1},\dots,z_{l+b-1})),\\
    B_2=&((z_{n-b+2},\dots,z_n,z_1),\dots,(z_n,z_1,\dots,z_{b-1})) \text{ and }\\
    B_3=&((z_{n-b+2},\dots,z_n,z_{l+1}),\dots,(z_{n},z_{l+1},\dots,z_{l+b-1})).
\end{align*}
    If  $z_i=0$ for $1\leq i\leq l$, then $w_H(B_1)+w_H(B_2)-w_H(B_3)\leq b-1.$
\end{lemma}
\begin{proof}
We prove this lemma by considering the following cases:
\begin{itemize}
    \item If $z_{l+1}\neq 0$, then $w_H(B_1)=b-1$ and $w_H(B_3)=b-1.$ Hence   $w_H(B_1)+w_H(B_2)-w_H(B_3)=w_H(B_2)\leq b-1.$
    \item  If $z_{l+1}=0$ and $z_{l+2}\neq 0$, then $w_H(B_1)=b-2$ and $w_H(B_3)\geq b-2.$ Hence   $w_H(B_1)+w_H(B_2)-w_H(B_3)\leq w_H(B_2)\leq b-1.$
  \item If $z_{l+1}=0,z_{l+2}=0,\dots, z_{l+(i-1)}=0$ and $z_{l+i}\neq 0$  for some $1\leq i\leq b-1$, then $w_H(B_1)=b-i$ and $w_H(B_3)\geq b-i.$ Hence   $w_H(B_1)+w_H(B_2)-w_H(B_3)\leq w_H(B_2)\leq b-1.$  
\end{itemize}
\end{proof}

\begin{lemma}\label{lemmab-1}
    Let $z=(z_1,z_2,\dots,z_n)\in \mathbb{F}_q^n$ and $b,l$ be integers such that $n \geq b\geq 2$ and $n \geq l\geq b.$ Consider
    \begin{align*}
    A_1=&((z_{l-b+2},z_{l-b+3},\dots,z_{l},z_{l+1}), \dots, (z_{l},z_{l+1},\dots, z_{l+b-1})),\\
    A_2=& ((z_{n-b+2},\dots,z_n,z_1),   \dots, (z_{n-1},z_n,z_1,\dots,z_{b-2}), (z_n,z_1,\dots,z_{b-1})),\\
    A_3=& ((z_{l-b+2},z_{l-b+3}, \dots,z_{l},z_1), (z_{l-b+3}, \dots,z_{l},z_1,z_2),\dots, (z_{l},z_1,z_2,\dots, z_{b-1}) ) \text{ and }\\
    A_4=& ((z_{n-b+2},z_{n-b+3},\dots,z_n,z_{l+1}), (z_{n-b+3},\dots,z_n,z_{l+1},z_{l+2}),\dots, (z_n,z_{l+1},z_{l+2},\dots, z_{l+b-1})),
\end{align*}
    then $w_H(A_1)+w_H(A_2)-w_H(A_3)-w_H(A_4)\leq b-1.$
\end{lemma}

\begin{proof}
    If $z_l\neq 0$ then $w_H(A_1)=b-1$ and $w_H(A_3)=b-1.$ Hence $w_H(A_1)+w_H(A_2)-w_H(A_3)-w_H(A_4)=w_H(A_2)-w_H(A_4)\leq b-1.$\par
    Next consider the case $z_l=0$ and $z_{l-1}\neq 0$, then we have $w_H(A_1)\geq b-2,$ (since the tuples $(z_{l-b+2},z_{l-b+3},\dots,z_l,z_{l+1}),\dots, (z_{l-1},z_l,z_{l+1},\dots,z_{l+b-2})$ are all nonzero and are $b-2$ in number). Similarly, $w_H(A_3)\geq b-2,$ since $(z_{l-b+2},z_{l-b+3},\dots,z_l,z_1),\dots, (z_{l-1},z_l,z_1,\dots,z_{b-2})$ are all nonzero. If $w_H(A_1)\leq w_H(A_3)$ then $$w_H(A_1)+w_H(A_2)-w_H(A_3)-w_H(A_4)\leq w_H(A_2)-w_H(A_4)\leq b-1.$$ Else, if $w_H(A_1)>w_H(A_3)$ that is $w_H(A_1)=b-1$  implies that at least one of $z_{l+1},\dots, z_{l+b-1}$ is nonzero. Consequently, $w_H(A_4)\geq 1.$ Hence $$w_H(A_1)+w_H(A_2)-w_H(A_3)-w_H(A_4)\leq b-1+w_H(A_2)-(b-2)-1=w_H(A_2)\leq b-1.$$~\\
    Next, assume that $z_l=0,z_{l-1}=0$ and $z_{l-2}\neq 0.$ Then similar to the above arguments, we have $w_H(A_1)\geq b-3$ and $w_H(A_3)\geq b-3.$ Again, if $w_H(A_1)\leq w_H(A_3)$ then  $w_H(A_1)+w_H(A_2)-w_H(A_3)-w_H(A_4)\leq b-1.$ Otherwise, if $w_H(A_1)>w_H(A_3)$ then $w_H(A_1)=b-3+s$ for $s\in \{1,2\}.$
    \begin{itemize}
        \item If $s=1$ then $z_{l+b-1}\neq 0$ and $z_{l+b-2},z_{l+b-3},\dots,z_{l+1}$ are zero. Consequently, $w_H(A_4)\geq 1.$ Hence $w_H(A_1)+w_H(A_2)-w_H(A_3)-w_H(A_4)\leq b-2+w_H(A_2)-(b-3)-1=w_H(A_2)\leq b-1.$
        \item If $s=2$ then $z_{l+j}\neq 0$ for some $1\leq j\leq b-2.$ Consequently, $w_H(A_4)\geq 2.$ Hence   $w_H(A_1)+w_H(A_2)-w_H(A_3)-w_H(A_4)\leq b-1+w_H(A_2)-(b-3)-2=w_H(A_2)\leq b-1.$
    \end{itemize}
    
    Now, assume that $z_l=0,z_{l-1}=0,\dots, z_{l-(i-1)}=0$ and $z_{l-i}\neq 0$ for $0\leq i\leq b-2.$ Then similar to the above arguments, we have $w_H(A_1)\geq b-1-i,$ since the tuples $$(z_{l-b+2},\dots,z_{l-i},0,0,\dots,0,z_{l+1}),\dots,
        (z_{l-i},0,0,\dots,z_{l+1},\dots,z_{l+b-i-1})$$
   are all nonzero.  Similarly, $w_H(A_3)\geq b-1-i,$ since the tuples
   \begin{align*}
        (z_{l-b+2},\dots,z_{l-i},0,0,\dots,0,z_1),\dots,
        (z_{l-i},0,0,\dots,z_{1},\dots,z_{b-i-1})
    \end{align*}
    are nonzero.
    The remaining tuples in $A_1$ are as follows
    \begin{align*}
      A_1'=((0,\dots,0,z_{l+1},\dots,z_{l+b-i}),\dots, (0,0,z_{l+1},\dots,z_{l+b-2}),(0,z_{l+1},\dots,z_{l+b-1}) ).
    \end{align*}
     Consider the following tuples from $A_4$
      \begin{align*}
        A_4'=((z_{n-i+1},z_{n-(i-1)+1},\dots,z_n,z_{l+1},\dots,z_{l+b-i}),\dots, (z_{n-1},z_n,z_{l+1},\dots,z_{l+b-2}),(z_n,z_{l+1},\dots,z_{l+b-1})).
    \end{align*}
    Again, if $w_H(A_1)\leq w_H(A_3)$ then we are done. Let $w_H(A_1)>w_H(A_3)$ then $w_H(A_1)=b-1-i+s$ for some $1\leq s\leq i.$ Therefore, exactly $s$ of the tuples in $A_1'$ are nonzero. Let $z_{l+j}\neq 0$ for some $1\leq j\leq b-1.$ Note that if $1\leq j\leq b-i$ then all tuples in $A_1'$ are nonzero, that is $w_H(A_1)=b-1$ and all tuples of $A_4'$ are also nonzero, that is $w_H(A_4)\geq i.$ Hence $$w_H(A_1)+w_H(A_2)-w_H(A_3)-w_H(A_4)\leq b-1-(b-1-i)+w_H(A_2)-i=w_H(A_2)\leq  b-1.$$
    
  For $j>b-i$, we have $s<i$, and all tuples in $A_4'$ are not necessarily nonzero. But at least $s$ tuples in $A_1'$ are nonzero. Corresponding to those tuples, at least $s$ tuples of $A_4'$  $$(z_{n-s+1},\dots,z_n,z_{l+1},\dots,z_{l+b-s}),\dots,(z_n,\dots,z_{l+b-1})$$
   are nonzero. Consequently, $w_H(A_4)\geq s.$ Hence  $$w_H(A_1)+w_H(A_2)-w_H(A_3)-w_H(A_4)\leq s+w_H(A_2)-s=w_H(A_2)\leq  b-1.$$   
Next, if $z_l=0,z_{l-1}=0,\dots,z_{l-b+2}=0$ then 
 \begin{align*}
     A_1=&((0,0,\dots,0,z_{l+1}),(0,0,\dots,0,z_{l+1},z_{l+2}),\dots,(0,z_{l+1},\dots,z_{l+b-1})),\\
     A_3=&((0,0,\dots,z_1),(0,0,\dots,0,z_1,z_2),\dots,(0,z_1,\dots,z_{b-1})).
 \end{align*}
 Again, if $z_{l+1}\neq0$ then $w_H(A_1)=b-1$ and $w_H(A_4)=b-1.$ Hence $$w_H(A_1)+w_H(A_2)-w_H(A_3)-w_H(A_4)\leq w_H(A_2)-w_H(A_3)\leq  b-1.$$ 
 If $z_{l+1}=0$ and $z_{l+2}\neq 0$  then $w_H(A_1)=b-2$ and $w_H(A_4)\geq b-2.$ Hence $$w_H(A_1)+w_H(A_2)-w_H(A_3)-w_H(A_4)\leq w_H(A_2)-w_H(A_3)\leq  b-1.$$  

 Now, assume that $z_{l+1}=0,z_{l+2}=0,\dots,z_{l+i-1}=0$ and $z_{l+i}\neq 0$ for $1\leq i \leq b-1,$ then $w_H(A_1)=b-i$ and $w_H(A_4)\geq b-i.$ Hence $$w_H(A_1)+w_H(A_2)-w_H(A_3)-w_H(A_4)\leq w_H(A_2)-w_H(A_3)\leq  b-1.$$ \par
 If $z_l=z_{l-1}=\dots=z_{l-b+2}=0$ and $z_{l+1}=z_{l+2}=\dots=z_{l+b-2}=0$. Then $A_1=0$. Hence, $w_{H}(A_1)+w_{H}(A_2)-w_H(A_3)-w_H(A_4)\leq w_H(A_2) \leq b-1.$ This completes the proof.
\end{proof}
The next result gives a recurrence relation between $A_b(n,d,q)$ and $A_b(n-m,d-2r-b+1,q)$. Sometimes, it may be easy to calculate $A_b(n-m,d-2r-b+1,q)$ that consequently gives a bound on $A_b(n,d,q)$.
\begin{theorem}\label{abnqd}
    Let $n$ and $d$ be positive integers with $2\leq d\leq n.$ If there are two integers $r$ and $m$ satisfying  $b\leq d-2r\leq n-m$ and $0\leq r\leq m$, then
    $$A_b(n,d,q)\leq \frac{q^m}{B_b(m,r)}A_b(n-m,d-2r-b+1,q).$$ 
\end{theorem}
\begin{proof}
 Let $C$ be $(n,M)$ code over $\mathbb{F}_q$ with $b$-symbol distance at least $d$ such that $M=A_b(n,d,q).$ 
 Define $\rho : \mathbb{F}_q^n \to \mathbb{F}_q^m$ and $\rho': \mathbb{F}_q^n \to \mathbb{F}_q^{n-m}$ as follows
 $$\rho(x_1,\dots,x_{n-m},x_{n-m+1},\dots,x_n)=(x_{n-m+1},\dots,x_n)$$
 and $$\rho'(x_1,\dots,x_{n-m},x_{n-m+1},\dots,x_n)=(x_1,\dots,x_{n-m}).$$
 Let 
     $C'=C \cap \rho^{-1}(B_b(m,r))
     =\{c=(c_1,\dots,c_n)\in C~|~ \rho(c)=(c_{n-m+1},\dots,c_n)\in B_b(m,r)\}.$ Take $C''=\rho'(C').$ Then $C''$ is a code of length $n-m$ over $\mathbb{F}_q.$
     
 We claim that the minimum $b$-symbol distance of $C''$ is greater than or equal to $d-2r-b+1$ that is $d_b(C'')\geq d-2r-b+1.$ In order to show this let $(x_1,\dots,x_{n-m})$ and $(y_1,\dots,y_{n-m})$ be distinct codewords in $C''.$ Then there are $x=(x_1,\dots,x_{n-m},x_{n-m+1},\dots,x_n)$ and $y=(y_1,\dots,y_{n-m},y_{n-m+1},\dots,y_n)$ in $C'$ such that $\rho'(x)=(x_1,\dots,x_{n-m})$ and $\rho'(y)=(y_1,\dots,y_{n-m}).$
Also, since $x,y\in C'$ therefore $\rho(x)=(x_{n-m+1},\dots,x_n),\rho(y)=(y_{n-m+1},\dots,y_n)\in B_b(m,r).$ By triangle inequality, we get $d_b(\rho(x),\rho(y))\leq d_b(\rho(x),0)+d_b(0,\rho(y))\leq r+r=2r.$ Consequently, we have 
\begin{equation}
    d-w_b((z_{n-m+1},z_{n-m+2},\dots,z_{n})) \geq d-2r\label{eq1}
\end{equation}
where $z_i=x_i-y_i.$ From now onwards, we write $z=x-y$ and $z_i=x_i-y_i$ for all $1\leq i\leq n.$ Also, since $C'\subseteq C$ therefore $d\leq d_b(C').$ Observe that 
\small\begin{align*}
    d\leq& d_b(x,y)=w_b(z)=w_H((z_1,z_2,\dots,z_b), \dots , (z_{n-m},z_{n-m+1},\dots,z_{n-m+b-1}), (z_{n-m+1},z_{n-m+2},\dots, z_{n-m+b})\\
    &\dots, (z_n,z_1,\dots,z_{b-1})) \\
    =& w_H((z_1,z_2,\dots,z_b), \dots, (z_{n-m-b+1}, z_{n-m-b+2},\dots, z_{n-m}),(z_{n-m-b+2},z_{n-m-b+2},\dots,z_{n-m},z_{n-m+1}),\\
    & \dots, (z_{n-m},z_{n-m+1},\dots, z_{n-m+b-1}), (z_{n-m+1},z_{n-m+2},\dots, z_{n-m+b}),\\
    &(z_{n-m+2},\dots, z_{n-m+b+1}), \dots, (z_{n-b+1},z_{n-b+2},\dots, z_n),(z_{n-b+2},\dots,z_n,z_1),\\
    &\dots, (z_{n-1},z_n,z_1,\dots,z_{b-2}), (z_n,z_1,\dots,z_{b-1}))\\
    =&w_H(S_1,S_3,S_2,S_4),
\end{align*}
where 
\begin{align*}
    S_1=& (z_1,z_2,\dots,z_b), \dots, (z_{n-m-b+1}, z_{n-m-b+2},\dots, z_{n-m}),\\
    S_2=&(z_{n-m+1},z_{n-m+2},\dots, z_{n-m+b}),
    (z_{n-m+2},\dots, z_{n-m+b+1}), \dots, (z_{n-b+1},z_{n-b+2},\dots, z_n),\\
    S_3=&(z_{n-m-b+2},z_{n-m-b+2},\dots,z_{n-m},z_{n-m+1}), \dots, (z_{n-m},z_{n-m+1},\dots, z_{n-m+b-1}), \text{ and }\\
    S_4=& (z_{n-b+2},\dots,z_n,z_1),   \dots, (z_{n-1},z_n,z_1,\dots,z_{b-2}), (z_n,z_1,\dots,z_{b-1}).\\
    \text{Let } S_5=& (z_{n-m-b+2},z_{n-m-b+3}, \dots,z_{n-m},z_1), (z_{n-m-b+3}, \dots,z_{n-m},z_1,z_2),\dots, (z_{n-m},z_1,z_2,\dots, z_{b-1}) \text{ and}\\
    S_6=& (z_{n-b+2},z_{n-b+3},\dots,z_n,z_{n-m+1}), \dots, (z_n,z_{n-m+1},z_{n-m+2},\dots, z_{n-m+b-1}).
\end{align*}
Note that $w_H(S_1,S_5)=w_b(z_1,z_2,\dots,z_{n-m})$ and $w_H(S_2,S_6)=w_b(z_{n-m+1},\dots,z_n).$
Thus, we have 
\begin{align*}
    d\leq &w_H(S_1,S_3,S_2,S_4)+w_H(S_5)+w_H(S_6)-w_H(S_5)-w_H(S_6)\\
    =& w_H(S_1)+w_H(S_3)+w_H(S_2)+w_H(S_4)+w_H(S_5)+w_H(S_6)-w_H(S_5)-w_H(S_6)\\
    =& w_H(S_1,S_5)+w_H(S_2,S_6)+w_H(S_3)+w_H(S_4)-w_H(S_5)-w_H(S_6)\\
    =& w_b(z_1,z_2,\dots,z_{n-m})+w_b(z_{n-m+1},\dots,z_n)+w_H(S_3)+w_H(S_4)-w_H(S_5)-w_H(S_6)\\
    d-&w_b(z_{n-m+1},\dots,z_n) \leq w_b(z_1,z_2,\dots,z_{n-m})+w_H(S_3)+w_H(S_4)-w_H(S_5)-w_H(S_6)
    \end{align*}
   \text{ Using Eq. \ref{eq1}, we get }$$d-2r\leq w_b(z_1,z_2,\dots,z_{n-m})+w_H(S_3)+w_H(S_4)-w_H(S_5)-w_H(S_6).$$ 

In Lemma \ref{lemmab-1}, take $l=n-m,A_1=S_3,A_2=S_4,A_3=S_5$ and $A_4=S_6.$ Then $$w_H(S_3)+w_H(S_4)-w_H(S_5)-w_H(S_6)\leq b-1.$$ We have $$d-2r\leq  w_b(z_1,z_2,\dots,z_{n-m})+b-1 $$ equivalently,
$$d-2r-b+1\leq w_b(z_1,z_2,\dots,z_{n-m}).$$ Hence, $d_b(C'')\geq d-2r-b+1.$ Thus $|C''|\leq A_b(n-m,d-2r-b+1,q)$.

Next, we claim that $|C'|=|C''|$ that is $\rho'$ is one-one map on $C'$ To see this, let $x,y\in C'=C \cap \rho^{-1}(B_b(m,r))$ be distinct codewords. This implies 
\begin{equation}\label{eq2}
0<d_b(\rho(x),\rho(y))\leq 2r, \text{ and }   d_b(x,y)\geq d_b(C)=d.
\end{equation}
Now, let  $\rho'(x)=\rho'(y)$ that is $(x_1,\dots,x_{n-m})=(y_1,\dots,y_{n-m})$ implies that $x_i-y_i=z_i=0$ for $i=1,\dots,n-m.$ Observe that
\begin{align*}
d_b(x,y)=& w_b(z)\\
=&w_H((z_1,\dots,z_b),\dots,(z_{n-m-b+1},\dots,z_{n-m}), (z_{n-m-b+2},\dots,z_{n-m},z_{n-m+1}),\\
&(z_{n-m-b+3},\dots,z_{n-m},z_{n-m+1},z_{n-m+2}),\dots, (z_{n-m},\dots, z_{n-m+b-1}),\\
&(z_{n-m+1},\dots,z_{n-m+b}), (z_{n-m+2},\dots,z_{n-m+b+1}),\dots, (z_n,z_1,\dots,z_{b-1})) \\
=&w_H((0,0,\dots,0,z_{n-m+1}),(0,0,\dots,0,z_{n-m+1},z_{n-m+2}),\dots, (0,z_{n-m+1},z_{n-m+2},\dots, z_{n-m+b-1})\\
&(z_{n-m+1},z_{n-m+2},\dots,z_{n-m+b}),\dots,(z_{n-b+1},\dots,z_n),(z_{n-b+2},\dots,z_n,0),\dots (z_{n-1},z_n,0,\dots,0),\\
&(z_n,0,\dots,0))  \\
=& w_H(A,B,C)=w_H(A)+w_H(B)+w_H(C),\\
\text{where }\ \ &\\
& A= (0,0,\dots,0,z_{n-m+1}),(0,0,\dots,0,z_{n-m+1},z_{n-m+2}),\dots, (0,z_{n-m+1},z_{n-m+2},\dots, z_{n-m+b-1}),\\
&B= (z_{n-m+1},z_{n-m+2},\dots,z_{n-m+b}),\dots,(z_{n-b+1},\dots,z_n), \text{ and}\\
&C=(z_{n-b+2},\dots,z_n,0), (z_{n-b+3},\dots,z_n,0,0),\dots,(z_n,0,\dots,0).\\
 \text{ Let } &D=(z_{n-b+2},\dots,z_n,z_{n-m+1}),(z_{n-b+3},\dots,z_n,z_{n-m+1},z_{n-m+2}),\dots,(z_n,z_{n-m+1},\dots,z_{n-m+b-1}).
 \end{align*}
 \text{ Then } $w_H(B)+w_H(D)=w_H(B,D)=w_b(z_{n-m+1},z_{n-m+2},\dots,z_n).$ Thus, we have
 \begin{align*}
   d_b(x,y)=& w_H(A,B,C)=w_H(A)+w_H(B)+w_H(C)+w_H(D)-w_H(D)\\
 d_b(x,y)=& w_H(B,D)+w_H(A)+w_H(C)-w_H(D)\\
 =&w_b(z_{n-m+1},z_{n-m+2},\dots,z_n)+w_H(A)+w_H(C)-w_H(D)\\
 =&d_b(\rho(x),\rho(y))+w_H(A)+w_H(C)-w_H(D).
 \end{align*}
In Lemma \ref{3b-1}, take  $l=n-m, B_1=A,B_2=C,B_3=D$ and  $z=(0,0,\dots,z_{n-m+1},\dots, z_n)$ then
$w_H(A)+w_H(C)-w_H(D)\leq b-1.$ So, we get 
$$d\leq  d_b(x,y)\leq  d_b(\rho(x),\rho(y))+(b-1)\leq 2r+b-1\ \text{(Using Eq. \ref{eq2})}$$ 
which is a contradiction since $d\geq 2r+b.$
Hence $\rho'(x)\neq \rho'(y)$ that is $\rho'$ is one-one consequently, $|C'|=|C''|.$

Our next claim is that $|C'|\geq \frac{MB_b(m,r)}{q^m}.$  To show; let  $\bm{c_1},\bm{c_2},\dots,\bm{c_M}$ be images of codewords of $C$ under the map $\rho.$ (need not be distinct) Then $|C'|$ is the number of $\bm{c}_i$'s that belong to $B_b(m,r)$ (since cardinality does not depend on the choice of the center). Choose $a\in \mathbb{F}_q^m$ as the center of a $b$-ball of radius $r$ in $\mathbb{F}_q^m$ such that the number of $\bm{c}_i$'s that belongs to $B_b(m,r,a)$ is as large as possible. Now, replace $B_b(m,r)$ by $B_b(m,r,a)$ that is $|C'|$ is the number of $\bm{c}_i$'s that belong to $B_b(m,r,a).$ For any $u\in \mathbb{F}_q^m,$ consider characteristic function \\ 
$$\chi_u(\bm{c_i})=\begin{Bmatrix}
    1 \text{ if } \bm{c_i}\in B_b(m,r,u),\\
    0 \text{ if } \bm{c_i}\notin B_b(m,r,u),
\end{Bmatrix}.$$
The number of $\bm{c}_i$'s that belong to the ball $B_b(m,r,u)$ can be expressed as $\sum_{i=1}^M \chi_u(\bm{c}_i).$ Finally, note that the number
 of $\bm{c}_i$'s that belong to $B_b(m,r,a)$ is larger than or equal to the average overall $b$-balls (this is because $a$ is chosen such that the
 number of $\bm{c_}i$'s belonging to $B_b(m,r,a)$ is as large as possible  when $a$ ranges over $\mathbb{F}_q^m$), that is,
 \begin{align*}
     |C'|\geq & \frac{1}{q^m}\sum_{u\in \mathbb{F}_q^m}\sum_{i=1}^M\chi_u(\bm{c_i})\\
     =&\frac{1}{q^m}\sum_{i=1}^M\sum_{u\in \mathbb{F}_q^m}\chi_u(\bm{c_i})\\
     =&\frac{MB_b(m,r)}{q^m},
 \end{align*}
  where the last equality holds because for any fixed $i$ with
 $1\leq i\leq M,
\sum_{u\in \mathbb{F}_q^m} \chi_u(\bm{c_i})=|B_b(m,r,\bm{c}_i)|=B_b(m,r).$
Therefore $\frac{MB_b(m,r)}{q^m}\leq |C'|=|C''|\leq A_b(n-m,d-2r-b+1,q)$ that is  $M=A_b(n,d,q)\leq \frac{q^m}{B_b(m,r)}A_b(n-m,d-2r-b+1,q)$
\end{proof}

\begin{corollary}\cite[Theorem 1]{chen2022}
  For $b=2,$ let $n$ and $d$ be positive integers with
 $2\leq d\leq n.$ If there are two integers $r$ and $m$ satisfying  $2\leq d-2r\leq n-m$ and $0\leq r\leq m,$
 then we have
$$ A_2(n,d,q) \leq \frac{q^m}{B_2(m,r)} A_2(n-m,d-2r-1,q).$$  
\end{corollary}
\begin{example}
Assume notations are the same as in Theorem \ref{abnqd}. For $m=1$ and $r=0$, we have 
$$A_3(11,10,2)\leq 2A_3(10,8,2).$$
By computer search, we have $A_3(10,8,2)=23$. Thus  $$A_3(11,10,2)\leq 46.$$ 
By Sphere-packing bound \cite{Song2018}, we have 
$$A_3(11,10,2)\leq \frac{2^{11}}{B_3(11,4)}=\frac{2^{11}}{23}\cong 89.04$$ Hence Theorem \ref{abnqd} gives a better upper bound. By computer search, we have $A_3(11,10,2)=12$.
\end{example}
\section{Bounds on code size for $b$-symbol distance}\label{additivebound}
In this section, we present several bounds on code size for $b$-symbol distance.  First, we establish Johnson bound for even minimum $b$-symbol distance. The proof follows a similar approach to that of the Johnson bound (\cite[Theorem 2.3.8]{pless2003}).

\begin{theorem}\label{jboundeven}
Let $n$ and $d$ be positive integers such that $d=2t+2$ for some $t\in \mathbb{N}$. Then
$$A_b(n,d,q)\leq \frac{q^n}{B_b(n,t)+\frac{S_b(n,t+1)}{A_b(n,d,t+1,q)}}.$$
\end{theorem}
\begin{proof}
Let $C$ be a code of size $M=A_b(n,d,q)$ with $d_b(C) \geq d$. For $y\in \mathbb{F}_q^n$, define $$d_b(C,y)=\min_{c\in C, \   c\neq y}d_b(c,y).$$
Denote $\Im=\{y\in \mathbb{F}_q^n|\ d_b(C,y)=t+1\}$. Since $d=2t+2$, $\bigcup_{c\in C} B_b(n,t,c)\cup \Im$ is a disjoint union. Thus, we have 
\begin{equation}\label{j1}
    MB_b(n,t)+|\Im|\leq q^n.
\end{equation}
Let $\aleph =\{(c,y)\in C\times \Im| \ d_b(c,y)=t+1\}$. For $y\in \Im$, let $C_y=\{c-y|\ c\in C, d_b(c,y)=t+1\}$. Observe that $wt_b(c-y)=t+1$ and $d_b(c-y,c'-y)=d_b(c,c')\geq d$ for $c,c'\in C$ with $d_b(c,y)=t+1=d_b(c'-y)$. Therefore, $C_y$ is a constant  $b$-symbol weight code of length $n$ with constant $b$-symbol weight $t+1$ and $d_b(C_y)\geq d$. Hence, we have
$$|C_y|\leq A_b(n,d,t+1,q).$$
Consequently,
\begin{equation}\label{j2}
   |\aleph|\leq |\Im|\cdot A_b(n,d,t+1,q). 
\end{equation}
To obtain a lower bound on $|\aleph|$, let $\aleph_c=\{y\in \Im|\ (c,y)\in \aleph\}$ for $c\in C$. We have 
\begin{equation}\label{j3}
|\aleph|=\sum_{c\in C}|\aleph_c|.
\end{equation}
As the cardinality of $b$-symbol sphere is independent of center, for a fix $c\in C$, we have 
\begin{equation*}
    S_b(n,t+1)=|S_b(n,t+1,c)|=|\{y\in \mathbb{F}_q^n|\ d_b(c,y)=t+1\}|.
\end{equation*}
For an arbitrary $c'\in C$ distinct from $c$ and $y\in S_b(n,t+1,c)$ ($d_b(c,y)=t+1$ implies $d_b(C,y) \leq t+1$), by triangle inequality
$$d\leq d_b(c',c)\leq d_b(c',y)+d_b(c,y)=d_b(c',y)+t+1.$$
That is, $d_b(c',y)\geq d-t-1=t+1$ for any $c'\in C$ distinct from $c$. So, $d_b(C,y)=t+1$. Hence $y \in \Im$. Thus, if $y\in S_b(n,t+1,c)$ then  $y\in \aleph_c$.  Therefore, \begin{equation}\label{j4}
    |\aleph_c|=S_b(n,t+1).
\end{equation}
Combining Eqs. \ref{j1}, \ref{j2}, \ref{j3} and \ref{j4}, we have
$$MB_b(n,t)+\frac{MS_b(n,t+1)}{A_b(n,d,t+1,q)}\leq q^n.$$
Equivalently, 
$$A_b(n,d,q)\leq \frac{q^n}{B_b(n,t)+\frac{S_b(n,t+1)}{A_b(n,d,t+1,q)}}.$$
This completes the proof.
\end{proof}
\begin{remark}
     A similar bound for $b=2$ has been obtained for binary and $q$-ary codes in \cite{Elishco2020} and \cite{chen2022}, respectively.
\end{remark}

\begin{remark}
    Corollary $4.5$ below gives an upper bound for the value of $A_b(n,d,t+1,q)$ and thus gives a closed-form formula in Theorem $4.1$.
\end{remark}

Now, we explore that any code over $\mathbb{F}_q$ in $b$-symbol distance can be viewed as a code over $\mathbb{F}_{q^b}$ in Hamming distance such that $b$-symbol distance and Hamming distance of the codes are same. Using this correspondence, we derive several bounds on the code size for $b$-symbol distance. 

Let $\mathbb{F}_{q^b}$ be an extension field of $\mathbb{F}_q$ of degree $b$, then  there exists an element $\gamma$ in $\mathbb{F}_{q^b}$ such that $\mathbb{F}_{q^b}=\mathbb{F}_q(\gamma)$. Consequently, there exists a monic irreducible polynomial over $\mathbb{F}_{q}$ of degree $b$ such that $f(\gamma)=0$. Therefore, any $x\in \mathbb{F}_{q^b}$, can be uniquely written as $a_1+a_2\gamma+\dots+a_{b}\gamma^{b-1},$ where $a_i\in \mathbb{F}_q,$ for $1\leq i\leq b.$ For $x=(x_1,x_2,\dots,x_n)\in \mathbb{F}_q^n,$ define $\delta(x):=x'=(x_1+x_2\gamma+\dots+x_{b}\gamma^{b-1}, \ x_2+x_3\gamma+\dots+x_{b+1}\gamma^{b-1},\  \dots,\  x_n+x_1\gamma+\dots+x_{b-1}\gamma^{b-1}).$ Clearly, $x'\in \mathbb{F}_{q^b}^n.$
\begin{proposition}\label{btoham}
For any $(n,M)$ code over $\mathbb{F}_q$ with minimum $b$-symbol distance $d$, 
$$C':=\{\delta(c)~:~c \in C\}$$ is an $(n,M)$ code over $\mathbb{F}_{q^b}$ with minimum Hamming distance $d$. If $c \in C$ has $b$-symbol weight $w$ then Hamming weight of $\delta(c) \in C'$ is $w$.
\end{proposition}
\begin{proof}
    For $x=(x_1,x_2,\dots,x_n)\in \mathbb{F}_q^n,$ define $\delta(x):=x'=(x_1+x_2\gamma+\dots+x_{b}\gamma^{b-1}, \ x_2+x_3\gamma+\dots+x_{b+1}\gamma^{b-1},\  \dots,\  x_n+x_1\gamma+\dots+x_{b-1}\gamma^{b-1}).$ Clearly, $x'\in \mathbb{F}_{q^b}^n.$ Then $C'=\{\delta(c)| c\in C\}$ is an $(n,M)$ code over $\mathbb{F}_{q^b}.$ Also, $wt_b(c)=wt_H(\delta(c)).$ Hence, the minimum Hamming distance of $C'$ is $d.$ 
\end{proof}

\begin{corollary}[Restricted Johnson Bound for $b$-symbol distance]\label{rjbound}
 Let $n$, $d$ and $w$ be positive integers. Then  $$A_b(n,d,w,q)\leq \left \lfloor \frac{d}{d-2w+\frac{q^bw^2}{(q^b-1)n}} \right \rfloor$$ 
   whenever the denominator is positive.
\end{corollary}
\begin{proof}
 Let $C$ be a code over $\mathbb{F}_q$ such that $|C|=A_b(n,d,w,q).$ From Proposition \ref{btoham}, $C'$ is a constant Hamming weight code over $\mathbb{F}_{q^b}$ of length $n$ and minimum Hamming distance at least $d$ whose all codewords have Hamming weight $w.$ Using  Restricted Johnson Bound for Hamming distance \cite[Theorem 2.3.4]{pless2003}, we have
$$A_b(n,d,w,q)=|C|=|C'|\leq \left \lfloor \frac{d}{d-2w+\frac{q^bw^2}{(q^b-1)n}} \right \rfloor.$$
 whenever $\left (d-2w+\frac{q^bw^2}{(q^b-1)n}\right )$ is positive.
\end{proof}

\begin{corollary}[Plotkin bound for $b$-symbol distance]\label{plotb}
    Let $q$ be a prime power, and $r,n$ be positive integers satisfying $2\leq d\leq n$ and $rn<d,$ where $r=1-q^{-b}.$ Then 
    $$A_b(n,d,q)\leq \left \lfloor \frac{d}{d-rn}\right \rfloor.$$
\end{corollary}
\begin{proof}
 Let $C$ be a code such that $A_b(n,d,q)=|C|.$ By Proposition \ref{btoham}, $C'$ is a code of length $n$ over $\mathbb{F}_{q^b}$ with minimum Hamming distance at least $d.$ By Plotkin bound for Hamming distance \cite[Theorem 2.2.1]{pless2003}, we have 
 $$A_b(n,d,q)=|C|=|C'|\leq \left \lfloor \frac{d}{d-rn}\right \rfloor.$$
\end{proof}
\begin{remark}
    In \cite{yang2016}, Yang et al. derived the Plotkin-like bound for $b$-symbol distance. We obtained the same bound with a very easy approach in the above corollary. 
\end{remark}
We give an example of a code that attains the Plotkin bound for the $b$-symbol distance.
\begin{example}
For $n=7,b=3, d_b=7$ and $q=2,$ we have $$A_3(7,7,2)\leq \left \lfloor \frac{d}{d-rn}\right \rfloor=\left \lfloor \frac{7}{7-\frac{7}{8}\cdot7}\right \rfloor=8.$$ Consider a code $C$ with length $7$ over $\mathbb{F}_2$ given by $$C:=\{1 1 1 1 1 1 1,
    0 1 0 0 0 1 1,
    1 0 1 0 0 0 1,
    1 1 0 1 0 0 0,
    0 1 1 0 1 0 0,
    0 0 1 1 0 1 0,
    1 0 0 0 1 1 0,
    0 0 0 1 1 0 1\}.$$ It is easy to see $|C|=8$ and $d_b(C)=7$. Thus, the bound in Corollary $4.6$ is attained.  
\end{example}

\begin{corollary}[Sphere packing type  bound for $b$-symbol distance]\label{spb} Let $n$ and $d$ be positive integer with $1\leq d\leq n$. Let $t=\lfloor \frac{d-1}{2}\rfloor$. Then
    $A_b(n,d,q)\leq \frac{q^{bn}}{\sum_{i=0}^t {{n}\choose {i}}(q^b-1)^i}. $
\end{corollary}
\begin{proof}
    The proof will follow using Proposition \ref{btoham} and Sphere packing bound \cite[Theorem 1.12.1]{pless2003} for the Hamming distance.
\end{proof}
\begin{corollary}[Gilbert-Varshamov type bound for $b$-symbol distance]\label{Gilbert}
   Let $n$ and $d$ be positive integer with $1\leq d\leq n$. Then $A_b(n,d,q)\geq \frac{q^{bn}}{\sum_{i=0}^{d-1}{{n}\choose {i}}(q^b-1)^i}$.
\end{corollary}
\begin{proof}
     The proof will follow using Proposition \ref{btoham}, and Gilbert bound \cite[Theorem 2.8.1]{pless2003} for the Hamming distance.
\end{proof}
\begin{example}\label{gvexamlple}
Let $b=3$, $n=10$ and $d=8$. Then by Corollary \ref{Gilbert}, we have
\begin{equation*}
A_3(10,8,2)\geq \frac{2^{30}}{\sum_{i=0}^{7}{{n}\choose {i}}(7)^i}=8.368    
\end{equation*}
Thus, $A_3(10,8,2)\geq 9$. While by \cite[Theorem 3.2]{Song2018}, $A_3(10,8,2)\geq 6$. Hence, Corollary \ref{Gilbert} gives a better lower bound on code size. By computer search, we have $A_3(10,8,2)=23$.
\end{example}
\begin{corollary}[Elias bound for $b$-symbol distance]\label{elias}
Let $r=1-q^{-b}$. Let $w\leq rn$ and  $w^2-2rnw + rnd >0$. Then
    $$A_b(n,d,q)\leq \frac{rnd}{w^2-2rnw + rnd}\cdot\frac{q^{bn}}{\sum_{i=0}^{w}{{n}\choose {i}}(q^b-1)^i}.$$
\end{corollary}
\begin{proof}
     The proof will follow using Proposition \ref{btoham}, and Elias bound \cite[Theorem 2.5.3]{pless2003} for the Hamming distance.
\end{proof}

\section{ $b$-symbol linear programming  bound}\label{slpbound}
Linear programming bound is the most powerful bound that gives an upper bound on code size. This bound, as indicated by its name, is derived by solving a linear optimization problem.  In \cite{Elishco2020}, the authors derived linear programming bound specially for symbol-pair codes ($b=2$). In this section, we establish linear programming bound for $b$-symbol codes for any arbitrary $b\geq 2$. Our approach draws from techniques used in the Hamming distance.  Notably, we demonstrate that when $b=1$, our result aligns with the classical linear programming bound for Hamming distance.

Let $C_s$ be a simplex code over $\mathbb{F}_q$ with parameters $\left [\frac{q^b-1}{q-1},b,d_H=q^{b-1}\right ]$ and generator matrix $G_s$. Define a map $\Phi:\mathbb{F}_q^b\to C_s$ as $\Phi(x)=x\cdot G_s$. It is easy to see that $\Phi$ is an $\mathbb{F}_q$-linear bijective map.

Let $C$ be an $(n,M,d_b)$ code over $\mathbb{F}_q$. For $c=(c_0,c_1,\dots,c_{n-1})\in C$, denote $\overline{c}_i=(c_i,c_{i+1},\dots, c_{i+b-1})$ (indices taken modulo $n$). Then $\pi(c)=(\overline{c}_0,\overline{c}_1,\dots,\overline{c}_{n-1})$. We concatenate code $C$ with $C_s$ as
\begin{equation*}
    C_H=\{\Phi(\pi(c))=(\Phi(\overline{c}_0),\Phi(\overline{c}_1),\dots,\Phi(\overline{c}_{n-1})) : \ c=(c_0,c_1,\dots,c_{n-1})\in C\}. \tag{$\ast$}
\end{equation*}

\begin{lemma}
    If $C$ is a code over $\mathbb{F}_q$ with parameters $(n,M,d_b)$ then $C_H$ is a code over $\mathbb{F}_q$ with parameters $\left ( \frac{(q^b-1)n}{q-1},M,d_H=q^{b-1}d_b\right )$.
\end{lemma}
\begin{proof}
    Since $\Phi$ is an $\mathbb{F}_q$-linear bijection, $C_H$ is a code of length $\frac{(q^b-1)n}{q-1}$ and cardinality $M$. Let $c\in C$ with $b$-symbol weight $wt_b(c)$. Then the Hamming weight of $(\Phi(\overline{c}_0),\Phi(\overline{c}_1),\dots,\Phi(\overline{c}_{n-1}))$ is $q^{b-1}wt_b(c)$,  since the simplex code $C_s$ is a constant weight code with constant weight $q^{b-1}$.  Thus, $C_H$ is a $\left ( \frac{(q^b-1)n}{q-1},M,d_H=q^{b-1}d_b\right )$ code.
\end{proof}
 Let $W_{C}^b(x)$ and $W_{C_H}^H(x)$ be the weight enumerator of $C$ with respect to the $b$-symbol distance and of $C_H$ with respect to the Hamming distance, respectively.
 \begin{lemma}\cite{Luo2024}
Let $C$ and $C_H$ be the codes as the above. Then
$$W_{C}^b(x)=1+A_1^bx+A_2^bx^2+\dots+A_n^bx^n$$
if and only if
$$W_{C_H}^H(x)=1+A_1^bx^{q^{b-1}}+A_2^bx^{2q^{b-1}}+\dots+A_n^bx^{nq^{b-1}}$$
 \end{lemma}

 Let $$B_i^b(C)=\frac{1}{M}|\{(u,v)\in C^2:\ d_b(u,v)=i\}$$
 and 
 $$B_i^H(C_H)=\frac{1}{M}|\{(x,y)\in C_H^2:\  d_H(x,y)=i\}  $$
 be distance distribution of $C$ with respect to $b$-symbol distance and of $C_H$ with respect to the Hamming distance, respectively.

\begin{lemma}\label{aeqb}
  \begin{enumerate}[(i)]
      \item 
     $B_j^H(C_H)= \begin{cases}
          B_i^b(C)\ \ \ \ \text{ if } \ \ j=iq^{b-1} \text{ for all } 1\leq i\leq n, \\
          0\ \ \ \ \ \ \ \ \ \ \ \ \  \text{ otherwise }
      \end{cases}$
      \item $\sum_{i=0}^n B_i^b(C)=M=\sum_{i=0}^N B_i^H(C_H)=\sum_{j=0}^n B_{jq^{b-1}}^H(C_H)$.
  \end{enumerate}  
\end{lemma}
\begin{proof}
$(i)$ Let $u,v\in C$. Then $\Phi(\pi(u)),\Phi(\pi(v))\in C_H$ and 
\begin{equation*}
    \begin{split}
d_H(\Phi(\pi(u)),\Phi(\pi(v)))&=w_H(\Phi(\pi(u))-\Phi(\pi(v)))\\
&=w_H(\Phi(\pi(u)-\pi(v))) \ \ \ \text{ (since $\Phi$ is linear)}\\
&=q^{b-1}w_H(\pi(u)-\pi(v))\\
&=q^{b-1}d_b(u,v).
    \end{split}
\end{equation*}  
$(ii)$ It will follow from the fact that $\Phi$ is a bijection.
\end{proof}
For any $q$,  Krawtchouck polynomial  $K^{q,n}_k(x)$  of degree $k$, $0 \leq k \leq n$ is defined as
$$K^{n,q}_k(x)=\sum_{j=0}^k(-1)^j (q-1)^{k-j} \binom{x}{j}\binom{n-x}{k-j}.$$

The following lemma is well-known in the Hamming distance (for instance,  see \cite[Lemma 2.6.2]{pless2003}.

\begin{lemma}\label{rootunity}
  Let $\xi$ be $q$-th root of unity in complex numbers. Let $u\in \mathbb{Z}_q^N$ with $w_H(u)=iq^{b-1}$. Then
\begin{equation*}
    K^{N,q}_k(iq^{b-1})=\sum_{v\in \mathbb{Z}_q^N, \  w_H(v)=k}\xi^{u\cdot v}.\tag{$\ast \ast$}
\end{equation*}
\end{lemma}

\begin{lemma}
Let $C$ be a code over $\mathbb{F}_q$ with parameters $(n,M,d_b)$ and $B_i^b(C)$, $0 \leq i \leq n$ be the distance distribution of $C$ with respect to $b$-symbol distance. Then 
    $$\sum_{i=0}^n B_i^b(C)K_k^{N,q}(iq^{b-1})\geq 0 \text{ for all integers } 0\leq k\leq n, \text{ where } N=\frac{(q^b-1)n}{q-1}$$
\end{lemma}
\begin{proof}
Let $C_H$ be the concatenated code. By Lemma \ref{aeqb}, we have 
\begin{equation*}
    \begin{split}
        M\sum_{i=0}^nB_i^b(C)K_k^{N,q}(iq^{b-1})=& M\sum_{i=0}^nB_{iq^{b-1}}^H(C_H)K_k^{N,q}(iq^{b-1})\\
        =&\sum_{i=0}^n\ \sum_{\substack{(x,y)\in C_H^2, \\ w_H(x-y)=iq^{b-1}}}\ \ \sum_{v\in \mathbb{Z}_q^N,\ w_H(v)=k} \xi^{(x-y)\cdot v} \text{ (by Lemma \ref{rootunity})}\\
        =&\sum_{(x,y)\in C_H^2}\sum_{\substack{v\in \mathbb{Z}_q^N,\\ w_H(v)=k}}\xi^{x\cdot v-y\cdot v} \\
        =&\sum_{\substack{v\in \mathbb{Z}_q^N,\\ w_H(v)=k}} \sum_{x\in C_H}\xi^{x\cdot v} \sum_{x\in C_H}\xi^{-x\cdot v}\\
        =& \sum_{\substack{v\in \mathbb{Z}_q^N,\\ w_H(v)=k}} |\sum_{x\in C_H}\xi^{x\cdot v}|^2\geq 0.
    \end{split}
\end{equation*}
This completes the proof.
\end{proof}

\begin{theorem}(Linear programming bound for $b$-symbol distance)\label{lpbound}
For given integers $q$, $n$, $d_b$ with $1\leq d_b\leq n$, we have
$$A_b(n,d,q)\leq \max \sum_{i=0}^n B_i^b$$
 subject to the constraints
 \begin{enumerate}[(1)]
     \item $B_0^b=1$, and $B_i^b=0$ for $1\leq i\leq d_b-1$.
     \item $B_i^b\geq 0$ for $d_b\leq i\leq n$.
     \item $\sum_{i=0}^n B_i^bK_k^{N,q}(iq^{b-1})\geq 0$  for all integers  $0\leq k\leq n$, where $N=\frac{(q^b-1)n}{q-1}$.
 \end{enumerate}
\end{theorem}
\begin{corollary}(Linear programming bound for the Hamming distance)
    For given integers $q$, $n$, $d_H$ ($b=1$) with $1\leq d_H\leq n$, we have
$$A_H(n,d,q)\leq \max \sum_{i=0}^n B_i$$
 subject to the constraints
 \begin{enumerate}[(1)]
     \item $B_0=1$, and $B_i=0$ for $1\leq i\leq d_H-1$.
     \item $B_i\geq 0$ for $d_H\leq i\leq n$.
     \item $\sum_{i=0}^n B_iK_k^{n,q}(i)\geq 0$  for all integers  $0\leq k\leq n$.
 \end{enumerate}
\end{corollary}

To illustrate the above linear programming bound, we provide an explicit example.
\begin{example}
Let $b=2, q=2, n=11$ and $d_b=10$. Then $N=33$ and $B_0^b=1$ and $B_i^b=0$ for $1\leq i\leq 9$. Thus linear programming problem is 
\begin{equation}
    \text{maximize}\ \  z=1+B_{10}^b+B_{11}^b
\end{equation}
subject to constraints
\begin{equation*}
\begin{split}
    B_{10}^b,B_{11}^b&\geq 0\\
    K_k^{33}(0)+ K_k^{33}(10\cdot 2)B_{10}^b+K_k^{33}(11\cdot 2)B_{11}^b&\geq 0\ \ \ \text{ for } 0\leq k\leq 11\\
    i.e., K_k^{33}(0)+ K_k^{33}(20)B_{10}^b+K_k^{33}(22)B_{11}^b&\geq 0\ \ \ \text{ for } 0\leq k\leq 11
  \end{split}  
\end{equation*}
By evaluating values of $K_k^{33}(0),K_k^{33}(20),K_k^{33}(22)$ for $0\leq k\leq 11$, we have 
\begin{equation*}
    \begin{split}
        1+ B_{10}^b + B_{11}^b&\geq 0\\
33 - 7B_{10}^b - 11B_{11}^b&\geq 0\\
528 + 8B_{10}^b + 44B_{11}^b&\geq 0\\
5456 + 56B_{10}^b - 44B_{11}^b&\geq 0\\
40920 - 160B_{10}^b - 220B_{11}^b&\geq 0\\
237336 - 112B_{10}^b + 748B_{11}^b&\geq 0\\
1107568 + 904B_{10}^b - 308B_{11}^b&\geq 0\\
4272048 - 456B_{10}^b - 2508B_{11}^b&\geq 0\\
13884156 - 2652B_{10}^b + 4488B_{11}^b&\geq 0\\
38567100 + 3380B_{10}^b + 1760B_{11}^b&\geq 0\\
92561040 + 4264B_{10}^b - 13156B_{11}^b&\geq 0\\
193536720 - 10088B_{10}^b + 9316B_{11}^b&\geq 0\\
    \end{split}
\end{equation*}
Upon solving this linear programming problem, we obtain $z\cong 5.71428$. Therefore, by LP bound, $A_2(11,10,2)\leq 5$. In comparison, sphere-packing bound (in \cite{Yuval2011}) provides $A_2(11,10,2)\leq 26$, while Johnson bound (in \cite{Elishco2020}) gives $A_2(11,10,2)\leq 16$. Also, an exhaustive computer search confirms that the exact value of $A_2(11,10,2)=4$.
\end{example}

\section*{Conclusion}
Codes in the Hamming distance are effective for communication channels that process one symbol at a time, such as binary symmetric channels, where errors typically affect individual symbols. These channels can correct errors by measuring the Hamming distance between codewords, which counts the number of positions where two codewords differ. However, there are channels that read multiple symbols simultaneously, such as in high-density data storage systems, where physical limitations or resolution constraints lead to errors affecting multiple symbols at once. In these cases, traditional Hamming distance codes are insufficient, as they are designed to handle errors in single symbols. To address this issue, the concept of $b$-symbol distance (with $b \geq 1$) has been introduced. When $b = 1$, this approach coincides with the traditional Hamming distance. By using b-symbol read channels, we can detect and correct errors that affect multiple symbols in these more complex channels, providing a more accurate error correction strategy than the traditional Hamming distance.

Establishing bounds on code size is a central challenge in coding theory. While numerous bounds have been derived for the Hamming distance, relatively few have been known for the $b$-symbol distance. In this paper, we extend the analysis of bounds for $b$-symbol distance and present several new upper and lower bounds. These include the linear programming bound, a recurrence relation on code size, the even Johnson bound, the restricted Johnson bound, the Gilbert-Varshamov-type bound, and the Elias bound for the $b$-symbol read channels with $b \geq 2$. Moreover, we provide examples that illustrate how the Gilbert-Varshamov bound we derive offers a stronger lower bound compared to the one presented in \cite{Song2018} (see Example \ref{gvexamlple}). We believe there is potential for further improvement in the sphere-packing and Gilbert-Varshamov bounds, and we identify this as a future direction. For practical applications, the next step is to construct codes that achieve these established bounds, ensuring both the reliability and efficiency of these codes. We encourage researchers to explore the construction of codes that meet the bounds presented in this paper.

\bibliographystyle{abbrv}
\bibliography{ref}

\end{document}